\documentclass[12pt]{amsart}

\usepackage{bbold}

\usepackage[top=1in,left=1in,right=1in,bottom=1in]{geometry}
\usepackage{subcaption}
\usepackage{amssymb}
\usepackage[dvipsnames]{xcolor}
\newcommand\myshade{85}
\colorlet{mylinkcolor}{violet}
\colorlet{mycitecolor}{YellowOrange}
\colorlet{myurlcolor}{Aquamarine}

\usepackage[unicode=true,pdfusetitle,
bookmarks=true,bookmarksnumbered=false,bookmarksopen=false,
breaklinks=false,pdfborder={0 0 1},backref=false,
linkcolor  = mylinkcolor!\myshade!black,
citecolor  = mycitecolor!\myshade!black,
urlcolor   = myurlcolor!\myshade!black,
colorlinks = true,
]
{hyperref}
\usepackage[nameinlink]{cleveref}

\usepackage{bbm}
\usepackage{tikz}
\usepackage{autonum} 
\usepackage{dsfont} 
\usepackage{extarrows} 
\usepackage{braket}

\usepackage{environ} 
\usepackage{floatrow} 
\usepackage{bm} 

\newcommand{\QQ}{{\mathcal{Q}}}

\newcommand{\R}{\mathbb{R}}

\newcommand{\Q}{\mathbb{Q}}
\newcommand{\N}{\mathbb{N}}
\newcommand{\NN}{\mathcal{N}}
\newcommand{\p}{{\partial}}

\newcommand{\dd}[2]{\dfrac{\partial #1}{\partial #2}}

\newcommand{\PP}{{\mathcal{P}}}

\newcommand{\Tr}{{\operatorname{Tr}}}
\newcommand{\lctr}{{\operatorname{Tr}}}

\newcommand{\Hall}{{\mathrm{Hall}}}
\newcommand{\spec}{{\mathrm{Spec}}}

\newcommand{\GG}{\mathcal{G}}
\newcommand{\tg}{{\tilde{g}}}

\newcommand{\supp}{\mathrm{supp}}

\newcommand{\Bb}{\mathbb{B}}

\newcommand{\Z}{\mathbb{Z}}
\newcommand{\G}{\mathbb{G}}
\newcommand{\Dd}{\mathbb{D}}

\newcommand{\Pp}{\mathbb{P}}
\newcommand{\bbLambda}{\mathbb{\Lambda}}

\newcommand{\CC}{\mathcal{C}}

\newcommand{\One}{\mathds{1}}

\newcommand{\II}{\mathcal{X}}

\newcommand{\MM}{\mathcal{M}}

\newcommand{\UU}{{\mathcal{U}}}
\newcommand{\KK}{{\mathcal{K}}}

\newcommand{\Hh}{{\mathbb{H}}}

\newcommand{\C}{\mathbb{C}}

\newcommand{\oz}{{\overline{z}}}

\renewcommand{\Im}{\operatorname{Im}}

\newcommand{\1}{\mathds{1}}
\newcommand{\RR}{\mathcal{R}}
\newcommand{\de}{ \ \mathrel{\stackrel{\makebox[0pt]{\mbox{\normalfont\tiny def}}}{=}} \ }

\numberwithin{equation}{section}

\usepackage[alphabetic,nobysame]{amsrefs}

\newcommand{\xz}{\color{JungleGreen}}

\newcommand{\PUV}{{\Psi_{U,V}}}

\newcommand{\bfx}{{\bm{x}}}

\newcommand{\bfz}{{\bm{z}}}

\newcommand{\MG}{{\mathcal G}}

\renewcommand{\Im}[1]{\operatorname{\mathbb{I}\mathbbm{m}}\{#1\}}

\title[Absolutely continuous spectrum for topological insulators]{Absolutely continuous spectrum for truncated topological insulators}
\author{Alexis Drouot}
\address[Alexis Drouot]{University of Washington, Seattle, USA.} 
\email{adrouot@uw.edu}

\author{Jacob Shapiro}
\address[Jacob Shapiro]{Princeton University, Princeton, USA} 
\email{jacobshapiro@princeton.edu}

\author{Xiaowen Zhu}
\address[Xiaowen Zhu]{University of Washington, Seattle, USA.} 
\email{xiaowenz@uw.edu}

\NewEnviron{equations}{%
\begin{equation}\begin{gathered}
  \BODY
\end{gathered}\end{equation}
}
\newtheorem{thm}{Theorem}

\newtheorem{definition}{Definition}
\newtheorem{lemma}{Lemma}[section]

\newtheorem{proposition}{Proposition}
\newtheorem{theorem}[thm]{Theorem}

\theoremstyle{definition}
\newtheorem{rmk}{Remark}[section]

\newcommand{\ii}{\operatorname{i}}
\newcommand{\findex}{\operatorname{ind}}

\newcommand{\eql}[1]{\begin{align}#1\end{align}}

\begin{document} 
\newcommand{\ee}{\operatorname{e}}
\newcommand\norm[1]{\left\lVert#1\right\rVert}

\begin{abstract} 
We show that if a topological insulator is truncated along a curve that separates the plane in two sufficiently large regions, then the edge system admits absolutely continuous spectrum. 
Our approach combines a recent version of the bulk-edge correspondence along curves that separates geometry and intrinsic conductance \cite{DZ24}, with a result about absolutely continuous spectrum for straight edges \cite{BW22}.
\end{abstract}

\maketitle

\section{Introduction}

Topological insulators are quantum materials that experimentally exhibit unique electronic properties. While insulating in their bulk, they conduct electricity along their boundaries. This results in  mobile edge states impervious to disorder \cite{Halperin1982_PhysRevB.25.2185,Aizenman_Graf_1998,EGS05} and has motivated the search for quasi-particles with exotic statistics \cite{KITAEV20062,Alicea2011,RevModPhys.80.1083}.  They are also a focus in quantum information theory because of their connections with the computational complexity of many-body systems \cite{KempeKitaevRegev2006}. This makes topological insulators a central theme in condensed matter physics.

Mathematically, this electronic behavior is explained by an index theorem called the bulk-edge correspondence \cite{Hatsugai,KS02}. It relates the topological index of an infinite (bulk) system with the index of its half-space truncation (edge). The edge index results from spontaneous currents running along the boundary of the sample, while the bulk index corresponds to electricity as a response to voltage. As a direct corollary of the bulk-edge correspondence, bulk spectral gaps are filled for topological insulators truncated to half-spaces. 

Because the resulting states experimentally correspond to \emph{currents}, the RAGE theorem \cite[Theorem 2.6]{AizenmanWarzel2016} suggests that the bulk gaps are filled by absolutely continuous spectrum. This is a rare occasion where one may establish the existence of \emph{extended states}, see Fr\"ohlich--Graf--Walcher \cite{Frohlich2000}, de Bi\'evre--Pul\'e \cite{DP02}, and Germinet--Klein \cite{Germinet_Klein_Schenker_2007} for Landau-like Hamiltonians and proofs via Moure estimates. In a recent paper, Bols--Werner \cite{BW22} extended these results to general tight-binding models by relying instead on the bulk-edge correspondence and an index theory result \cite{ABJ20}. 



The first proof of the bulk-edge correspondence dates back to \cite{Hatsugai}; see also \cite{KS02,EG02}. The scenario where one works with a curved edge instead of a straight edge has been investigated only recently \cite{LT22,L23, DZ23, DZ24}. In analogy with \cite{BW22}, it is natural to ask whether absolutely continuous spectrum emerges in bulk spectral gaps of such \textit{curved} edge systems. In this note, we show that if the truncation region and its complement contain parabolas, then the bulk gap of the edge system is filled with absolutely continuous spectrum. Our analysis combines the Bols--Werner approach with a novel version of the bulk-edge correspondence in curved settings \cite{DZ24}, which decorrelates the geometry of the model to the intrinsic characteristic of the bulk.

\subsection{Setup and main result.} 
We consider the motion of electrons in a two-dimensional material in the single-particle and discrete space approximations. In this setting, the dynamics of one electron are governed by a Hamiltonian--a bounded self-adjoint operator $H$ on $\ell^2(\Z^2;\C^m)$ that is local:

\begin{definition}[Local operator] $H$ is \emph{local} if there exists $\nu>0$ such that its 
kernel satisfies
\[
    \norm{H_{xy}} \leq \nu^{-1} \ee^{-\nu\norm{x-y}}, \qquad x,y\in\Z^2.
\]
\end{definition}

We will work here with \textit{insulators:}

 \begin{definition}[Spectrally-gaped insulator] A local Hamiltonian $H$ models an insulating system at the Fermi energy $E_F\in\R$ if \eql{\label{eq:spectral gap condition} E_F\notin \spec(H) \,. } 
 \end{definition}

We mention that \eqref{eq:spectral gap condition} is merely sufficient to define an inulator: for instance, a mobility gap can also give rise to insulators in a physically meaningful way; see \cite[Equations (1.2),(1.3)]{EGS05}. 
Insulators have vanishing \emph{longitudinal} conductivity. However, their transversal conductivity -- the \textit{Hall conductivity} -- can be non-zero. Without translation invariance, it is given (at zero temperature) by the Kubo formula:

\begin{definition} The Hall conductance of a local Hamiltonian $H$, insulating at energy $E_F$ is
\begin{equation}
    \sigma_\mathrm{Hall}(H) := -\ii \Tr \big( P \big[ [P,\Lambda_2], [P,\Lambda_1] \big] \big)\,,
\end{equation}
    where $P$ denote the spectral projection of $H$ below energy $E_F$ and $\Lambda_1, \Lambda_2$ are the characteristic functions of $\{ X_1 > 0 \}$ and $\{ X_2 > 0\}$, respectively.
\end{definition}

When the insulating bulk system is truncated to a subset or when two \emph{different} bulk systems are brought together along an edge, the insulating condition \eqref{eq:spectral gap condition} may break down. The system may then support spontaneous currents along the edge. The present work is an investigation of this phenomenon.

\subsection{Truncation of bulk systems}

To model an interface system, we choose a subset $U\subset \R^2$ whose boundary $\partial U$ will be the interface between two different bulk systems. Precisely,

\begin{definition}[interface Hamiltonian]\label{def:DW edge Hamiltonian} Let $H_\pm$ be two local insulating bulk Hamiltonians at the same energy $E_F \in \R$ and $U\subset \R^2$. An interface Hamiltonian compatible with $H_\pm$ and $U$ is a \emph{local} self-adjoint bounded operator $\widehat{H}$ such that for some $\nu>0$, \eql{\label{eq:interface edfge} \norm{E_{xy}}\leq\nu^{-1} e^{-\nu d\left(x,\partial U\right)}, \qquad x,y\in\Z^2,  } where $E := \widehat{H}- \Lambda_U H_+ \Lambda_U - \Lambda_{U^c} H_- \Lambda_{U^c}$. 
\end{definition}

\Cref{def:DW edge Hamiltonian} is related to another model of edge systems (see e.g. \cite{EG02,FSSWY20}): working with the Hilbert space $\ell^2(\Z\times\N,\C^m)$ and asking that far away from the edge $\Z\times \{0\}$ the bulk and the edge operators agree. 

\subsection{Main result}
The goal of this note is to show that under a mild condition on $U$, an interface Hamiltonian interpolating between different topological phases has absolutely continuous spectrum in the bulk gap. Hence, by the RAGE theorem, it behaves as a conductor. 

\begin{definition}\label{def:1} A parabolic region is a set of the form
\eql{
    \big\{ My \in \R^2 :  \ y_1 \geq 0, \ |y_2| \leq y_1^\alpha \big\},
}
where $\alpha > 0$ and $M$ is a rigid motion of $\R^2$ (the composition of a translation and a rotation).
\end{definition}

We are ready for our main result:

\begin{theorem}\label{thm:existence of ac spectrum for edge} Assume that:
\begin{itemize}
    \item[(a)] $H_\pm$ are two local Hamiltonians with a joint spectral gap $\GG\subseteq\R$, and distinct Hall conductance at energies in this gap.
    \item[(b)] $\widehat{H}$ is an interface edge Hamiltonian compatible with $H_\pm$ and $U$ (see \Cref{def:DW edge Hamiltonian}).
    \item[(c)] $U \subset \R^2$ is such that both $U$ and $U^c$ contain parabolic regions.
\end{itemize}
Then the absolutely continuous spectrum of $\widehat{H}$ contains $\GG$. 
\end{theorem}


The condition (c), see \Cref{fig:0}, is stronger than the one given in \cite{DZ23}, but also yields a stronger result. There we merely asked that $U$ and $U^c$ contain arbitrarily large balls instead of parabolic regions and we proved that $\GG \subset \spec(\widehat{H})$, 
with no mention of spectral type.

\begin{figure}[b] 
\caption{\label{fig:0} The set $U$ (in blue) satisfies the assumptions of \Cref{thm:existence of ac spectrum for edge}: both $U$ and $U^c$ contain the parabolic regions $\PP_+$ and $\PP_-$. If $U$ and $U^c$ are filled by different topological phases, \Cref{thm:existence of ac spectrum for edge} predicts that $\p U$ support asymmetric edge states.} 
    \includegraphics[scale=1]{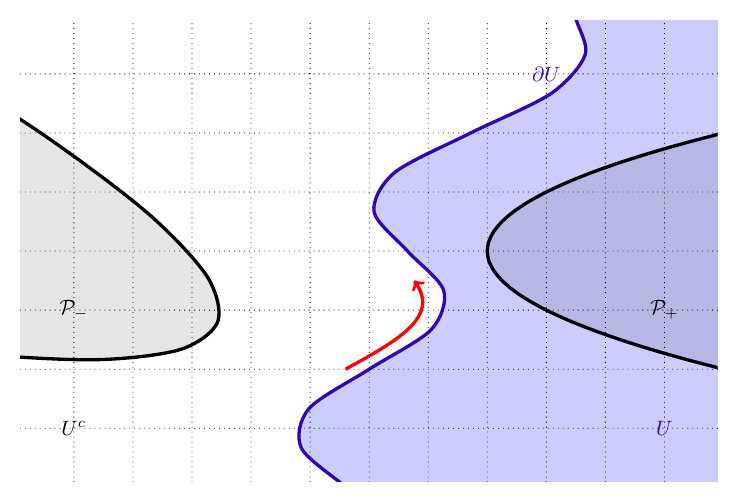}  
\end{figure}

\subsection{Strategy.} 

To explain our strategy let us delve into topological indices.

\subsubsection{The bulk topological indices}\label{subsubsec:bulk topology}
For insulating bulk systems, there is a well-defined topological index which in physics is usually referred to as the \emph{Chern number} \cite{Hasan_Kane_2010}. It is associated to the Hall conductance of the integer quantum Hall effect, i.e., to a linear response of the system to the application of electric voltage. As such, topologically the Chern number is usually written (in our setting which does not enjoy translation invariance) \cite{ASS94,FSSWY20} as \eql{\label{eq:Chern number} \NN(H) = \findex(\Pp U) = \findex(\bbLambda_1 \exp(2 \pi \ii \Lambda_2 P \Lambda_2 )) } where we use the short-hand notation $\Q W \equiv Q W Q + Q^\perp$ for a projection $Q$ and a unitary $W$; and $P \equiv \chi_{(-\infty,E_F)(H)}$ is the Fermi projection associated to $H$ (at Fermi energy $E_F\in\R$--i.e., all indices depend also on the choice of $E_F$) and $\Lambda_j \equiv \chi_{(0,\infty)}(X_j)$ for $j=1,2$ is the switch function along the $1,2$ axis; $U\equiv\exp(\ii \arg(X_1+\ii x_2))$ is the so-called Laughlin flux-insertion. We shall refer to the first expression in \eqref{eq:Chern number} as the Laughlin flux-insertion formula \cite{BvES94} and to the second as Kitaev's formula \cite{KITAEV20062}. We are also using the fact that if $Q$ is a self-adjoint projection and $W$ is a unitary, such that $[Q,W]\in\KK$ then $\Q W$ is a Fredholm operator. Local Hamiltonians which are spectrally-gapped at $E_F\in\R$ have a local Fermi projection $P$, in which case $[P,U]\in\KK$ (see e.g. \cite[Lemma A.1]{Bols2023}). In this same circumstance on $H$, it is also known \cite{FSSWY20} that $[\Lambda_1,\exp(2 \pi \ii \Lambda_2 P \Lambda_2 )) ]\in\KK$, so that both indices are well-defined. The fact that both are equal (without recourse to the double-commutator trace formula, but strictly via homotopy of operators) may be found in \cite[Lemma 4.7]{Bols2023}.

Considered from the perspective of linear response theory, the Chern number above is known to be equal to the Hall conductivity given by the Kubo formula \eql{\label{eq:double-commutator-formula} \NN(H) = - 2\pi\sigma_{\mathrm{Hall}}(H) = 2\pi \ii \lctr\left(P [[\Lambda_1,P],[\Lambda_2,P]]\right)\,.} It is known that when $P$ is local, the above double-commutator is indeed trace-class \cite{EGS05}.

\subsubsection{Edge systems and their topology}
We wish to define what it means to be an edge system on an independent footing \emph{without reference} to a pre-existing bulk system following \cite[Def. 3.4]{SW22a} or \cite[Def. 2.8]{Bols2023}. It stands to reason that taking a direct consequence of \Cref{def:DW edge Hamiltonian} which guarantees topological indices are well-defined is appropriate.
\begin{definition}[Abstract edge systems]\label{def:edge system} Let $\G\subset\Z^2$ (we have in mind either $\G=\Z^2$ for interface or $\G=U$ for sharp edges), $V\subset\G$ be two infinite sets and $E_F\in\R$ be given. Then the local Hamiltonian $\widehat{H}$ on $\ell^2(\G,\C^m)$ is said to be an edge Hamiltonian which descends from an insulating bulk system, with edge transverse to $\partial V$, iff there exists some smooth $g:\R\to[0,1]$ such that $\supp(g')$ is contained in an open interval which contains $E_F$, $g(-\infty)=1$ and $g(\infty)=0$, and such that \eql{\label{eq:edge index well defined}\left[\exp\left(2\pi\ii g(\widehat{H})\right),\Lambda_{V}\right]\in\KK} where $\KK$ is the ideal of compact operators, $\Lambda_{V}\equiv\chi_{V}(X)$ with $X$ the position operator on $\ell^2(\G)$, and $\chi$ is the characteristic function.
\end{definition}
The typical scenario usually considered in the literature (prior to \cite{DZ23,DZ24}) is truncation of a bulk system to the upper half plane and $V$ is taken as the upper right quadrant. It is well-known \cite{EG02,FSSWY20} that if one takes a spectrally gapped bulk system(s) and truncates as in \Cref{def:DW edge Hamiltonian} then \eqref{eq:edge index well defined} holds. In what follows we will be interested in deriving geometric conditions on the truncation shape which yield \eqref{eq:edge index well defined}.

The discussion in \eqref{subsubsec:bulk topology} has analogs for edge systems $\widehat{H}$ which obey \Cref{def:edge system} with $V\subset \R^2$. Indeed, the analogous index formula \cite{KS02} is \eql{\label{eq:edge index formula} \widehat{\NN}(\widehat{H},V) = \findex\left(\bbLambda_V \exp\left(2\pi\ii g(\widehat{H})\right)\right)\,.} There is also an analogous trace-formula which more readily corresponds to conductivity \cite[Eq-n (1.7)]{EGS05}:
\eql{\label{eq:edge trace formula}
\widehat{\sigma^V_{\mathrm{Hall}}}\left(\widehat{H}\right) = \lctr\left(g'\left(\widehat{H}\right)\ii\left[\Lambda_V,\widehat{H}\right]\right)
} where $g$ has the same meaning as above. The equality $2\pi\widehat{\sigma_{\mathrm{Hall}}}\left(\widehat{H}\right) = \widehat{\NN}(\widehat{H},V)$ has been first shown in \cite{KS02} for straight edges, but the same proof essentially follows in our case as we show below in \Cref{prop:edge index thm}.

It is clear, following \cite{BW22}, that 
\begin{theorem}[Bols-Werner]\label{thm:BW}
    If $\widehat{H}$ obeys \Cref{def:edge system} (with compact replaced by trace class) for some $V\subset\R^2$ and $\widehat{\NN}(\widehat{H},V)\neq 0$ then there exists some $\varepsilon>0$ such that \eql{\label{eq:ac spectrum} \left(E_F-\varepsilon,E_F+\varepsilon\right)\subset\spec_{\mathrm{ac}}(\widehat{H})\,.} 
\end{theorem}


This theorem relies on \cite[Theorem 2.1 3.]{ABJ20}, which states that if $W$ is a unitary and $Q$ is a projection so that $[W,Q]$ is trace-class (so $QWQ+Q^\perp$ is Fredholm) and so that $\findex\left(QWQ+Q^\perp\right)\neq0$ then $\spec_{\mathrm{ac}}(W)=\mathbb{S}^1$; cf. \cite[Theorem 3.1]{BDF_1973}.

With these preliminaries we are ready for the

\begin{proof}[Proof of \Cref{thm:existence of ac spectrum for edge}] We first construct a set $V \subset \R^2$ such that the edge Hamiltonian $\widehat{H}$ from item (b) of \Cref{thm:existence of ac spectrum for edge} has a non-zero edge topological index w.r.t. $V$, i.e., \eqref{eq:edge index formula} is non-zero. It being nonzero relies on two facts: that the two Chern numbers associated with $H_\pm$ are different, as well as a special geometric property of $U$ and $V$. This is explained below in \Cref{prop:existence of transversal set}. Our proof of this relies on the framework developed in \cite{DZ24} for the bulk-edge correspondence along curved edges. Specifically, we construct $V$ so that (i) $V$ is transverse to $U$ (this essentially says that $\p U$ and $\p V$ get further away from each other at a sufficiently fast rate) and (ii) $\p V$ is the range of a simple curve $\gamma : \R \rightarrow \R^2$ that starts in $U^c$ and ends in $U$. Under these conditions, \cite[Theorem 1]{DZ24} predicts that
\eql{\label{eq:BEC}
    \widehat{\NN}\left(\widehat{H},V\right) = 2\pi\chi_{U,V}\left( \sigma_{\mathrm{Hall}}(H_+)-\sigma_{\mathrm{Hall}}(H_-)\right)
} for some $\chi_{U,V}$ which depends on the geometry of $U$ and $V$ but not on $H_\pm$, and $V$ is precisely constructed so $\chi_{U,V}$ is non-zero.

Hence the right hand side of \eqref{eq:BEC} is non-zero by hypothesis (a) of \Cref{thm:existence of ac spectrum for edge}, and hence by \Cref{thm:BW} we have \eqref{eq:ac spectrum}.
\end{proof}

\begin{rmk}
    In fact \cite[Theorem 1]{DZ24} does not quite show \eqref{eq:BEC}, but rather it shows \eql{\label{eq:BEC-trace} \widehat{\sigma^V_{\mathrm{Hall}}}\left(\widehat{H}\right) = \chi_{U,V} \left(\sigma_{\mathrm{Hall}}(H_+)-\sigma_{\mathrm{Hall}}(H_-)\right)\,.
} To bridge the gap we provide two alternatives here. In the first, we show a curved edge index theorem: \eql{\label{eq:curved edge index thm}2\pi\widehat{\sigma^V_{\mathrm{Hall}}}\left(\widehat{H}\right) = \widehat{\NN}\left(\widehat{H},V\right)} so that we may invoke \cite[Theorem 1]{DZ24}. This is done in \Cref{prop:edge index thm} below.

In the latter scheme, we avoid $\widehat{\sigma^V_{\mathrm{Hall}}}\left(\widehat{H}\right)$ all together by proving an index theorem for the bulk Hall conductivity first (this is \Cref{lem:geometric inex thm}) and following a similar route for the bulk-edge correspondence as in \cite{FSSWY20}. This is done in \Cref{thm:geometric index BEC} below.
\end{rmk}

The rest of this paper (after some remarks) is devoted to the proof of \Cref{prop:existence of transversal set} as well as the two alternatives described above.

\subsection{Remarks}  While our main result proves that the bulk gap is filled with ac-spectrum, it cannot exclude superposition with other spectral types. For instance, if $\widehat{H}$ is an edge Hamiltonian whose associated bulk is gapped at $E_F\in\RR$ and $\varphi \in \ell^2(\Z^2,\C^m)$ then 
\begin{equation}
    \widehat{H} + E_F \varphi \otimes \varphi^\ast
\end{equation}
is also an edge Hamiltonian, however $E_F \in \spec_{\mathrm{pp}}(\widehat{H}) $. Previous approaches via Mourre estimates would fail here since they derive \emph{pure} ac-spectrum.


\Cref{thm:existence of ac spectrum for edge} can be seen as the final step of the program started in \cite{DZ23}, where some of us showed (under a weaker assumption on $U$) that $(E_F-\varepsilon,E_F+\varepsilon)\cap\spec(\widehat{H})\neq\varnothing$ for some $\varepsilon>0$. The proofs, however, are quite different: \cite{DZ23} uses mainly the local character of the bulk index while the current work relies on stronger results: the bulk-edge correspondence \cite{DZ24} and the spectrum of unitary whose index with a projector is non-zero \cite{ABJ20}.


\subsection{Open problems}
\subsubsection{Big disks versus big parabolas} The condition given here in \Cref{def:1} for parabolic regions vs. the one given in \cite{DZ23} makes one wonder what the spectral type is when $U$ and $U^c$ contain arbitrarily large disks.

\subsubsection{The spectral type of the edge system in the mobility gap regime}
A much more delicate question, not adressed here, arises in the scenario where the bulk system is an insulator due to Anderson localization (in which case the spectral gap closes but a dynamical so-called \emph{mobility gap} arises \cite{EGS05}). In this case, the bulk Hamiltonian already has its gap filled with Anderson localized states and it is not entirely clear what would be the resulting spectral type of the associated edge system: do the bulk Anderson-localized states become resonances embedded within the absolutely continuous spectrum? This question is perpendicular to the present study.
\subsubsection{The Fu-Kane-Mele index}
For systems of the Fermionic time-reversal invariant class, i.e., class AII in the Altland-Zirnbauer classification, the Chern number is always zero. This happens when there is an anti-unitary operator (the time-reversal operator) $\Theta$ which has the property that $\Theta^2=-\One$. A Hamiltonian $H$ is then termed time-reversal symmetric iff $[H,\Theta]=0$, in which case it is easy to show that $\NN(H) = 0$. However, for such systems one has the Fu-Kane-Mele $\Z_2$ index \cite{Kane_Mele_2005,Fu_Kane_2007}. For systems without translation-invariance it is most conveniently phrased via the Atiyah-Singer $\Z_2$-valued half Fredholm index (\cite{Atiyah1969,SB_2015,FSSWY20}). In short, for Fredholm operators which obey $F = -\Theta F^\ast \Theta$ (we shall call them $\Theta$-odd) we always have $\findex(F) = 0$, but we may still define \eql{\label{eq:Atiyah-Singer half Fredholm} \findex_2(F) := \left[\left(\dim \ker F\right)\mod 2\right] \in\Z_2} and it is a fact that this quantity is norm-continuous and compactly-stable under perturbations which respect the $\Theta$-odd constraint. It thus turns out that if $[\Theta,H]=0$ then $F := \Pp U$ is $\Theta$-odd. Similarly also $\bbLambda_1 \exp\left(2\pi\ii\Lambda_2 P \Lambda_2\right)$ is. Hence the Fu-Kane-Mele index is given \cite{SB_2015} by \eql{\label{eq:FKM index} \NN_2(H) = \findex_2(\Pp U) = \findex_2(\bbLambda_1 \exp\left(2\pi\ii\Lambda_2 P \Lambda_2\right))\,.} We are still unaware of trace-formulas or linear response theory for the $\Z_2$ Fu-Kane-Mele index, but see \cite[Section 6.3]{Bols2023}.

It is thus natural to ask whether a non-zero bulk Fu-Kane-Mele index implies ac-spectrum for the edge with a curved boundary, especially given the recent result in \cite{BC24}. We postpone the resolution of this question to future work.


\subsection{Acknowledgement} We gratefully acknowledge support from the National Science Foundation DMS 2054589 (AD) and the Pacific Institute for the Mathematical Sciences (XZ). The contents of this work are solely the responsibility of the authors and do not necessarily represent the official views of PIMS.

\section{Curved boundaries and the geometric Hall conductance}
The following definition is taken from \cite[Definition 4]{DZ24}:
\begin{definition}[Transverse sets]\label{def:transverse}
    We say that two sets $U,V \subset \R^2$ are transverse if
\eql{\label{eq-0a}
    \liminf_{\norm{x}\to \infty} \frac{\log \PUV(x)}{\log \norm{x}} >0, \qquad \PUV(x) \de 1 + d(x,\p U) + d(x,\p V)\qquad(x\in\R^2).
}
\end{definition}
We note that in \cite{DZ24} we expressed transversality using the $\|\cdot\|_1$-norm, but because norms on $\R^2$ are all equivalent, we could as well have used the Euclidean distance.

For tranverse sets $U$ and $V$, we introduced in \cite{DZ24} an integer $\chi_{U,V}$ that, roughly speaking, computes how many times $\p U$ (oriented so that $U$ lies to the left of $\p U$) enters $V$; see \eqref{fig:5}. We will give an intrinsically geometric definition of $\chi_{U,V}$ shortly, but first we want to provide context for how it arises in the setting of calculating the Hall conductance of the integer quantum Hall effect.

The formula \eqref{eq:double-commutator-formula} corresponds to the linear response of the system response of the system in the following manner: one turns on an electric field along the $x_1$ axis and measures current along the $x_2$ axis. What if instead we turn on an electric field along the curve $\partial U$ and measure the current along the curve $\partial V$? Since the two are transversal as in the above definition, it stands to reason to define the Hall conductivity corresponding to this experiment as
\eql{\label{eq:geometric Hall conductance}
    \sigma_{\mathrm{Hall}}^{U,V}(H) \equiv  \ii \lctr\left(P [[\Lambda_U,P],[\Lambda_V,P]]\right)
} 
where $\Lambda_S \equiv \chi_S(X)$ for any set $S\subset\R^2$ with $X$ the position operator and $\chi$ the characteristic function.

It turns out that \eqref{eq:double-commutator-formula} and \eqref{eq:geometric Hall conductance} are not always equal. Consider the simple example that $U$ is not the upper half plane but rather a horizontal strip of finite width with $V$ still the right half plane. Then we expect two opposite-direction currents on the two components of $\partial U$ which cancel so that in that case \eqref{eq:geometric Hall conductance} should yield zero. It turns out that to relate \eqref{eq:geometric Hall conductance} and \eqref{eq:double-commutator-formula} one uses $\chi_{U,V}$, the times $\partial U$ intersects $\partial V$. Indeed, in \cite[Theorem 3]{DZ24} it was shown that
\eql{\label{eq:magic formula}
    \sigma_{\mathrm{Hall}}^{U,V}(H)  = \chi_{U,V}  \cdot\sigma_{\mathrm{Hall}}^{1,2}(H) \,,
}
when $U, V$ are transversal. Though unsatisfactory, one \emph{could} take \eqref{eq:magic formula} as one possible definition for $\chi_{U,V}$. To do so it is comforting to have a-priorily \Cref{lem:double commutator is trace-class} below which allows us to conclude that $\sigma_{\mathrm{Hall}}^{U,V}(H)$ is well-defined as soon as $U$ and $V$ are transverse.

Let us pause for a moment on the geometric definition of $\chi_{U,V}$, which we will use below, based on the following observations: 
\begin{itemize}
    \item[(i)] There exist a set $\UU \subset \R^2$ transverse to $V$ such that $U \cap \Z^2 = \UU \cap \Z^2$, whose unbounded boundary components are the ranges of countably many proper simple curves $\gamma_j:\R \rightarrow \R^2$, such that $\UU$ lies to the left of $\gamma_j$ -- see \cite[Lemma 7.1]{DZ24}.
    \item[(ii)] Second, using that $\UU$ and $V$ are transverse, the limits 
    \begin{equation}
        \lim_{t \rightarrow +\infty} \big(\Lambda_V \circ \gamma_j(t) - \Lambda_V \circ \gamma_j(-t) \big)
    \end{equation}
    are all well-defined. They are non-zero for finitely many $j$; see \cite[Lemma 7.5-7.7]{DZ24}. 
    \item[(iii)] The sum 
    \begin{equation}
        \chi_{U,V} \de \sum_j \lim_{t \rightarrow +\infty} \big(\Lambda_V \circ \gamma_j(t) - \Lambda_V \circ \gamma_j(-t) \big)
    \end{equation}
    is finite and independent of the set $\UU$ from (a). This is a consequence of $\Lambda_U = \Lambda_\UU$ and of \eqref{eq:magic formula}.    
\end{itemize}

\begin{figure}[b]
     \centering
     \begin{subfigure}[ht]{0.3\textwidth}
         \centering
         \includegraphics[width=\textwidth]{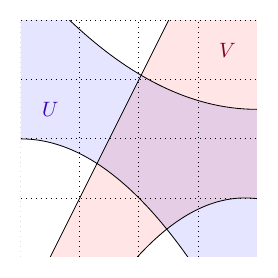}
         \caption{}
         \label{fig: 1(a)}
     \end{subfigure}
     \hfill
     \begin{subfigure}[ht]{0.3\textwidth}
         \centering
         \includegraphics[width=\textwidth]{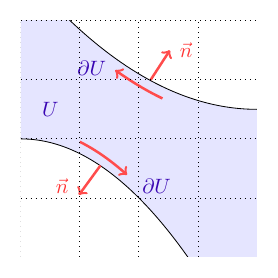}
         \caption{}
         \label{fig: 1(b)}
     \end{subfigure}
     \hfill
     \begin{subfigure}[ht]{0.3\textwidth}
         \centering
         \includegraphics[width=\textwidth]{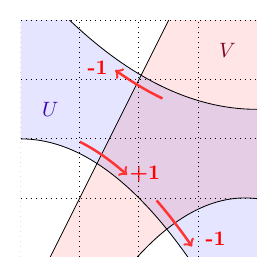}
         \caption{}
         \label{fig: 1(c)}
     \end{subfigure}
        \caption{We define the intersection number $\mathcal X_{U, V}$ between transverse simple sets $U, V$ in two steps. We first orient $\p U$ such that $U$ is to its left according to the outward-pointing normal, see (a) and (b); then we count how many times the oriented $\p U$ enters $V$, see (c). Here $\mathcal X_{U, V} = +1 - 1 - 1 = -1$.}
        \label{fig:5}
\end{figure}

Our main insight here is
\begin{proposition}\label{prop:existence of transversal set} Assume that $U$ and $U^c$ both contain parabolic regions. Then there exists some $V \subset \R^2$ such that $U$ and $V$ are transverse and with which $\chi_{U,V} \neq 0$.
\end{proposition}

\subsection{A geometric index formula}
It is sometimes convenient to work with indices of Fredholm operators rather than trace formulas, in order to eventually transition to the edge. To do so, we define the analog of the Kitaev index from \eqref{eq:Chern number}, for curved boundaries: 
\eql{
    \NN^{U,V}(H) := \findex\left(\bbLambda_V \exp\left(2\pi\ii\Lambda_U P \Lambda_U\right)\right)\,.
}

This index is well-defined for a similar reason as for the reason that $\sigma_{\mathrm{Hall}}^{U,V}(H)$ is well-defined: it follows from the fact $U$ and $V$ are transversal, as shown in \Cref{lem:geometric index well defined} below.

As a result, thanks to \eqref{eq:magic formula}, we have the analogous identity at the level of indices: \eql{
    \NN^{U,V}(H) = \chi_{U,V} \NN^{1,2}(H)\,.
}
It follows from two index theorems, the first, well known, states that $\NN^{1,2}(H) = 2\pi \sigma_{\mathrm{Hall}}^{1,2}(H)$. Its geometric analog follows a similar proof which is presented in \Cref{sec:geometric index theorem} below.



\section{Existence of a transversal set}\label{sec:existence of transversal set}
\begin{proof}[Proof of \Cref{prop:existence of transversal set}] The proof will be divided into five main steps: 

\textbf{Step 1.} (Definition of $V$) For notation purposes, define $U_+ = U$ and $U_- = U^c$; let $\alpha_\pm > 0$ and $M_\pm$ be rigid motions of $\R^2$ such that
\begin{equation}\label{eq-0b}
   \{ M_\pm y \in \R^2 : \ |y_2| \leq y_1^{\alpha_\pm}, \  y_1 \geq 0\} \subset U_\pm.
\end{equation}
Without loss of generalities, we can assume $\alpha_+ = \alpha_- \in (0,1/2]$ and denote it by $\alpha$: for instance, one can replace $\alpha_\pm$ by $\min(\alpha_+,\alpha_-, 1/2)$ while keeping \eqref{eq-0b} valid.  Let $y_\pm = M_\pm(0)$ and define:
\begin{equation}
    \Gamma_\pm = M_\pm (\R^+ \times \{0\}), \qquad \Gamma_0 = [y_-,y_+] \qquad \Gamma = \Gamma_- \cup \Gamma_0 \cup \Gamma_+.
\end{equation}
Here $[y_-, y_+]$ refers to the segment connecting $y_\pm$. See \Cref{fig:1}.

We now use terminology introduced in \cite[\S5]{DZ24}. Let $\gamma : \R \rightarrow \R^2$ be an arclength parametrization of $\Gamma$ with $\gamma(t) \in \Gamma_-$ for $t$ sufficiently small and $\gamma(t) \in \Gamma_+$ for $t$ sufficiently large. Note that $\gamma$ is a simple path with range $\Gamma$ -- see \cite[Definition 7]{DZ24}. Let $V$ be the connected component of $\R^2 \setminus \Gamma$ that lies to the left of $\gamma$ (see \cite[Proposition 7]{DZ24} for the rigorous definition of ``to the left of a simple path''); it has boundary $\Gamma$. Assuming for now that $U$ and $V$ are transverse, we compute $\chi_{U,V}$. Note that we have $\chi_{U,V} = -\chi_{V,U}$ because the left-hand-side of \eqref{eq:magic formula} changes sign when switching $U$ and $V$. Moreover, $V$ is a simple set, i.e. its boundary is the range of a simple path (see \cite[Definition 8]{DZ24}) and therefore, because $V$ lies to the left of $\gamma$, we have by \cite[Definition 9]{DZ24}:
\begin{equation}
    \chi_{V,U} = \lim_{t \rightarrow +\infty} \Lambda_U\big( \gamma(t) \big) -  \lim_{t \rightarrow -\infty} \Lambda_U\big( \gamma(t) \big) = 1.
\end{equation}
In the last equality, we used that for $t$ large, $\gamma(t) \in \Gamma_+ \subset U_+=U$ and for $t$ small, $\gamma(t) \in \Gamma_- \subset U_- = U^c$, so the two above limits are $1$ and $0$, respectively. Therefore, $\chi_{U,V} = -\chi_{V,U} = -1 \neq 0$.

\begin{figure}[b] 
\caption{\label{fig:1} In blue, the set $U$; in gray, the parabolic regions $\PP_+$ and $\PP_-$; in red, the set $V$ with its boundary $\Gamma$. One notes that $\chi_{U,V} = 1$. } 
    \includegraphics[scale=1]{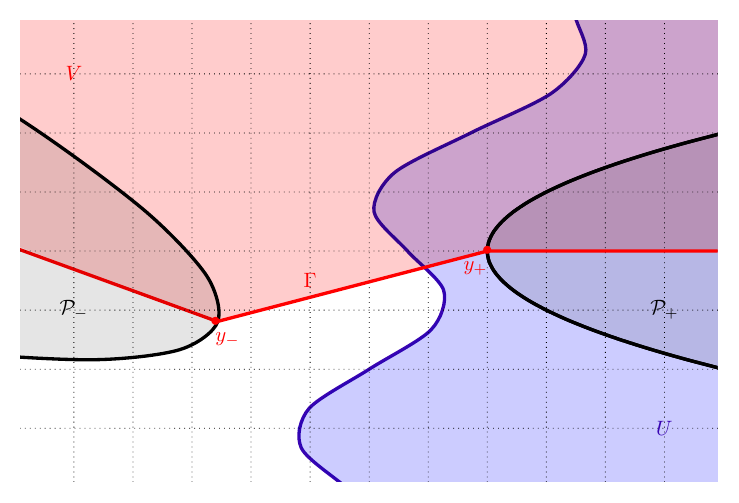}  
\end{figure}

\textbf{Step 2.} (Proof that $U$ and $V$ are transverse) Let $R \geq 5 + 2^{1+1/\alpha}$ such that $\Gamma_0 \subset \Dd_R(0)$. To prove that $U$ and $V$ are transverse, it suffices to prove the following inequality:
\begin{equation}\label{eq:0f}
    |x| \geq 2R \quad \Rightarrow \quad \Psi_{U,V}(x) \geq 2^{-2-\alpha} |x|.
\end{equation}
Let us introduce the sets:
\begin{equation}
    \QQ_\pm = \{ M_\pm y : \ y_1 \geq 0, \ 2|y_2| \leq y_1^\alpha \}.
\end{equation}
In Step 3 we prove the following inequalities:
\begin{align}
\label{eq:0c}   x \notin \Dd_{2R}(0) & \Rightarrow d(x,\Gamma_0) \geq 2^{-2-\alpha} \|x\|, \\
\label{eq:0d}   x \in \QQ_\pm^c \setminus \Dd_{2R}(0) & \Rightarrow d(x,\Gamma_\pm) \geq 2^{-2-\alpha} \|x\|, \\
\label{eq:0e}   x \in \QQ_\pm \setminus \Dd_{2R}(0) & \Rightarrow d(x,\p U) \geq 2^{-2-\alpha} \|x\|. 
\end{align}
Assume for now that these inequalities hold and fix $x \notin \Dd_{2R}(0)$. If $x \notin \QQ_+ \cup \QQ_-$, then by \eqref{eq:0c} and \eqref{eq:0d}:
\begin{equation}
    \Psi_{U,V}(x) \geq d(x,\p V) = \min_{\epsilon \in \{0,+,-\}} d(x,\Gamma_\epsilon) \geq 2^{-2-\alpha} \|x\|
\end{equation}
If $x \in \QQ_\pm$, then by \eqref{eq:0e}:
\begin{equation}
     \Psi_{U,V}(x) \geq d(x,\p U) \geq 2^{-2-\alpha} \|x\|
\end{equation}
This implies \eqref{eq:0f}.

\textbf{Step 3.} We prove \eqref{eq:0c}. Assume that $x \notin \Dd_{2R}(0)$. Then
\begin{equation}
    d(x,\Gamma_0) \geq d(x,\Dd_R(0)) = |x|-R \geq \dfrac{|x|}{2} \geq 2^{-2-\alpha} \|x\|
\end{equation}

\textbf{Step 4.} We prove \eqref{eq:0d}. Assume that $x \in \QQ^c_\pm \setminus \Dd_{2R}(0)$. Let $y \in \R^2$ with $M_\pm y = x$; 
remark for use below that
\begin{align}
  \dfrac{\| x \|}{2} + R \leq \| x \| = \| M_+ y \| \leq \| y \| + \| y_+ \| \leq \| y \| + R, \\
  2\| x\| -R \leq \| x \| = \| M_+ y \| \geq \| y \| - \| y_+ \| \geq \| y \| - R,
\end{align}
so $\|x \|/2 \leq \| y \| \leq 2\| x\|$. We have
\begin{equation}\label{eq:0g}
    d(x,\Gamma_\pm) = d(M_\pm y, M_\pm (\R^+ \times \{0\})) = d(y,\R_\pm \times \{0\}) \geq d(y,\R \times \{0\}) = |y_2|. 
\end{equation}
Moreover, because $x \notin \QQ_\pm$, we have $2|y_2| \geq |y_1|^\alpha$. In particular, $|y_2| \geq 1$: otherwise we would have $|y_1| \leq 2^{1/\alpha}$, so $\| y \| \leq 1 + 2^{1/\alpha}$, $\| x \| \leq 2 + 2^{1+1/\alpha}$, which is incompatible with $R \geq 5+2^{1+1/\alpha}$. From $2|y_2| \geq |y_1|^\alpha$, $|y_2| \geq 1$, and the inequality $a^p+b^p \geq (a+b)^p$, valid for $p = 2\alpha \leq 1$, we obtain:
\begin{equation}
    4|y_2| \geq 2|y_2|+|y_1|^\alpha \geq |y_2|^\alpha + |y_1|^\alpha \geq \| y\|^{\alpha}.
\end{equation}
Returning to \eqref{eq:0g}, we conclude that
\begin{equation}
    d(x,\Gamma_\pm) \geq \dfrac{\| y\|^{\alpha}}{4} \geq 2^{-2-\alpha} \| y\|^{\alpha}.
\end{equation}

\textbf{Step 5.} We finally prove \eqref{eq:0e}. Assume now that $x \in \QQ_\pm \setminus \Dd_{2R}(0)$; let $x_U \in \p U$ such that $\| x-x_U \| = d(x, \p U)$. We claim first that $[x,x_U]$ intersects $\p \PP_\pm$. Indeed if it did not, then because $x \in \QQ_\pm \subset \PP_\pm$, $x_U$ would be in $\PP_\pm$, which does not intersect $\p U$: contradiction. Let $x_0 \in [x,x_U] \cap \PP_\pm$. We have:
\begin{equation}
    d(x,\p U) = \| x-x_U\| \geq \| x-x_0 \| \geq d(x,\p \PP_\pm) = d(x,\PP_\pm^c).
\end{equation}
Let $x_* \in \p \PP_\pm$ with $d(x,\p \PP_\pm) = \| x-x_*\|$; let $y, y_*$ with $M_+ y = x$, $M_+ y_* = x_*$. 

Let $\CC_\pm$ be the region depicted in \eqref{fig:2}. Because $M_\pm^{-1} \PP_\pm$ is convex and the extremal points of $\CC_\pm$ are in $\PP_\pm$, we have $\CC_\pm \subset \MM_\pm^{-1} \PP_\pm$. It follows that
\begin{equation}
    d(x,\PP_\pm^c) = d(y,M_\pm^{-1} \PP_\pm^c) \geq d(y, \CC_\pm^c) = d(y,\p \CC_\pm).
\end{equation}
where in the last line we used that $y \notin \CC_\pm^c$.

\begin{figure}[t] 
\floatbox[{\capbeside\thisfloatsetup{capbesideposition={right,center},capbesidewidth=0.55\textwidth}}]{figure}[\FBwidth]
    {\caption{\label{fig:2} In black, the parabolic region $M_+^{-1}\PP_+$; in blue, the set $\CC_\pm$. We note that $d(x,\p U) = d(y,M_+^{-1}\PP_+) \geq d(y,\CC_\pm) = c$. }} 
    {\includegraphics[width=1\linewidth]{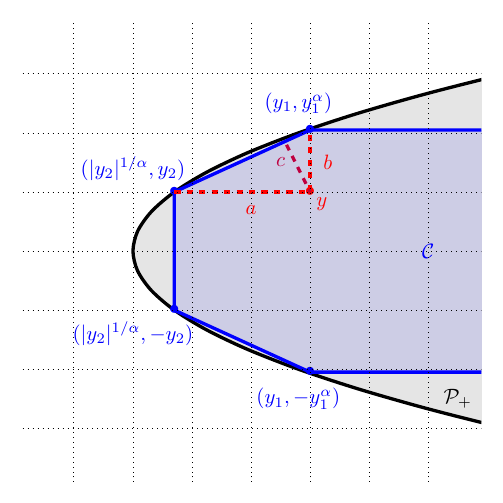} }
\end{figure}

With the notations of \eqref{fig:2}, we moreover have:
\begin{equation}
    d(y,\p \CC_\pm) = c = \dfrac{ab}{\sqrt{a^2+b^2}}, \qquad a = y_1-|y_2|^{1/\alpha}, \ b = y_1^\alpha - |y_2|,
\end{equation}
where we obtained the formula for $c$ via an area argument. Now because $x \in \QQ_\pm$, we have $y_1^\alpha \geq 2|y_2|$. So, $2b = 2y_1^\alpha - 2|y_2| \geq y_1^\alpha$ and $a \geq (1-2^{-1/\alpha}) y_1$. Moreover $a \leq y_1$ and $b \leq y_1^\alpha$. 
\begin{equation}\label{eq:0i}
    2c \geq (1-2^{-1/\alpha}) \dfrac{y_1^{1+\alpha}}{\sqrt{y_1^2 + y_1^{2\alpha}}}.
    \end{equation}
Remark moreover that $y_1 \geq 1$: otherwise we would have $2|y_2| \leq 1$ and hence $\| y \| \leq 2$, $\|x \| \leq 4$ which is incompatible with $\| x \| \geq R \geq 5$. It follows that $y_1 \geq \| y\|/2$:
\begin{equation}
    \| y \|^2 = y_1^2 + y_2^2 \leq y_1^2 + \dfrac{y_1^{2\alpha}}{4} \leq 4 y_1^2,
\end{equation}
Plugging these two inequalities in \eqref{eq:0i}, we obtain
\begin{equation}
    2c \geq (1-2^{-1/\alpha}) \dfrac{y_1^{1+\alpha}}{\sqrt{y_1^2 + y_1^{2\alpha}}} \geq \dfrac{1-2^{-1/\alpha}}{\sqrt{2}} y_1^\alpha \geq \dfrac{1}{4} \| y\|^\alpha \geq  2^{-2-\alpha} \|x\|.
\end{equation}
This completes the proof of \eqref{eq:0e}, and hence of \Cref{prop:existence of transversal set}. \end{proof}


\section{A curved edge index theorem}
In this section we prove \eqref{eq:curved edge index thm}.

\begin{proposition}\label{prop:edge index thm} The commutator $[e^{2i\pi \rho(\widehat{H})},\Lambda_V]$ is trace-class and
\begin{equation}\label{eq:0r}
    2\pi i \cdot \widehat{\sigma_\mathrm{Hall}^{V}}(\widehat{H}) = 
    \findex\big( \bbLambda_V \ee^{2\ii\pi \rho(\widehat{H})}  \big).
\end{equation}
\end{proposition}

\begin{proof} In the proof, we will use the following result, which corresponds to \cite[(3.10)]{DZ24}: if $f \in C_0^\infty(\GG)$, then for any $N \in \N$, there exists $C_N > 0$ such that 
\begin{equation}\label{eq:0t}
   \big| f(\widehat{H})(x,y) \big| \leq C_N e^{-N d_\ell(x,y) - N d_\ell(y,\p U) - N d_\ell(x, \p U)}, \qquad x,y \in \Z^2.
\end{equation}
In \eqref{eq:0t}, $f(\widehat{H})(x,y)$ denotes the kernel of $f(\widehat{H})$ at $(x,y)$ and $d_\ell$ refers to the logarithmic distance: $d_\ell(x,y) = \log (1+\|x-y\|)$.

\textbf{1.} Let $g(\lambda) = e^{2i\pi \rho(\lambda)} - 1$, which is a smooth function with support in the bulk spectral gap $\GG$. We have
\begin{equation}
    [e^{2i\pi \rho(\widehat{H})},\Lambda_V]  = [g(\widehat{H}),\Lambda_V].
\end{equation}
Because $g$ is smooth and supported in $\GG$, it satisfies the decay bound \eqref{eq:0t}. By \cite[(2.4)]{DZ24}, when $N$ is large enough, 
\[
    \left|[g(\widehat{H}), \Lambda_V](x,y)\right|\leq C_{3N} e^{-Nd_l(x, y) -Nd_l(x, \p V) - Nd_l(y, \p V) - Nd_l(x,\p U) - Nd_l(y, \p V)}.
\]
Since $(U, V)$ are transversal, by \cite[Corollary 2.4]{DZ24}, $\left[g(\widehat(H)), \Lambda_V\right]$ is trace class. 

\textbf{2.} We now focus on proving the index equality \eqref{eq:0r}. The original proof goes to \cite[Theorem 3.1]{KS02}; alternatively, we present here a direct approach, that relies on
\begin{align}
    \findex(\bbLambda_V \ee^{2i\pi \rho(\widehat{H})}) & = \Tr \big( e^{-2i\pi \rho(\widehat{H})} [g(\widehat{H}),\Lambda_V] \big) 
    \\
    &= \Tr \big( e^{-2i\pi \rho(\widehat{H})} g'(\widehat{H}) [\widehat{H},\Lambda_V] \big)
    = 2i\pi \cdot \Tr \big( \rho'(\widehat{H}) [\widehat{H},\Lambda_V] \big) = 2i\pi \cdot \sigma_e^{U,V}(\widehat{H}).
\end{align}
The first inequality follows from \cite{ASS94}, the last is a definition; the third equality comes from $e^{-2i\pi \rho} g' = 2i\pi \rho'$. The technical part consists of justifying the second equality. 

Let $I \subset \GG$ be an open set containing $\supp(g)$, and $\phi \in C^\infty(\R)$ such that $\phi = 1$ on $I^c$ and $\phi = 0$ on $\supp(g)$. We have:   
\begin{align}
   \findex(\bbLambda_V \ee^{2i\pi \rho(\widehat{H})}) & = \Tr \big( e^{-2i\pi \rho(\widehat{H})} [e^{2i\pi \rho(\widehat{H})},\Lambda_V] \big) 
    \\
    & = \Tr \big( e^{-2i\pi \rho(\widehat{H})} \big(1-\phi(\widehat{H})^2\big) [g(\widehat{H}),\Lambda_V] \big) \label{eq:0p} \\
    & + \Tr \big( e^{-2i\pi \rho(\widehat{H})} \phi(\widehat{H})^2 [g(\widehat{H}),\Lambda_V] \big). \label{eq:0o}
\end{align}
The trace \eqref{eq:0o} vanishes. Indeed, because $[g(\widehat{H}),1_V]$ is trace-class, we can use cyclicity \cite[(3.6)]{DZ24} to move one of the operators $\phi(\widehat{H})$ around. Using that $\phi g = 0$, we deduce that:
\begin{equation}
    \Tr \big( e^{-2i\pi \rho(\widehat{H})} \phi(\widehat{H})^2 [g(\widehat{H}),\Lambda_V] \big) = \Tr \big( e^{-2i\pi \rho(\widehat{H})} \phi(\widehat{H}) [g(\widehat{H}),\Lambda_V] \phi(\widehat{H})\big) = 0.
\end{equation}
We now treat with the term \eqref{eq:0p}. Let $\chi = e^{2i\pi \rho} (1-\phi^2)$, which is a smooth function supported in $\GG$. We claim that 
\begin{equation}\label{eq:0q}
    \Tr \big( \chi(\widehat{H}) [g(\widehat{H}),\Lambda_V] \big) = \Tr \big(  g'(\widehat{H}) \chi(\widehat{H})[\widehat{H},\Lambda_V] \big) = 2i\pi \cdot \Tr \big( \rho'(\widehat{H}) [\widehat{H},\Lambda_V] \big).
\end{equation}
To get the second equality in \eqref{eq:0q}, we simply observe that because $\phi$ vanishes on $\supp(g)$, $\chi g' = e^{2i\pi \rho} g' = 2i\pi\rho'$. We now focus on proving the first equality in \eqref{eq:0q}.

Let $\tg$ be an almost analytic extension of $g$. By the Helffer--Sj\"ostrand formula (with $dm(z) = \frac{dzd\oz}{\pi i}$),
\begin{align}
    \chi(\widehat{H})  [g(\widehat{H}),\Lambda_V] & = \chi(\widehat{H}) \left[ \int_{\C^+} \dd{\tg(z,\oz)}{\oz}  (\widehat{H}-z)^{-1} dm(z), \Lambda_V \right] \\
    & = \chi(\widehat{H}) \int_{\C^+} \dd{\tg(z,\oz)}{\oz}  \left[ (\widehat{H}-z)^{-1} , \Lambda_V \right] dm(z)
    \\
    & = \int_{\C^+} \dd{\tg(z,\oz)}{\oz} (\widehat{H}-z)^{-1} \chi(\widehat{H}) \left[ \widehat{H}  , \Lambda_V \right] (\widehat{H}-z)^{-1} dm(z). \label{eq:0s}
\end{align}

We now observe that the operator $\chi(\widehat{H}) \left[ \widehat{H}  , \Lambda_V \right]$ is trace-class. Indeed, $\chi$ is supported in $\MG$ so $\chi(\widehat{H})$ satisfies the decay bound \eqref{eq:0t}. By \cite[(2.4))]{DZ24}, $[\widehat{H},1_V]$ satisfies 
\begin{equation}
    [\widehat{H},\Lambda_V](x,y) \leq C_N e^{-N d_\ell(x,y) - N d_\ell(y,\p V) - N d_\ell(x, \p V)}, \qquad x,y \in \Z^2.
\end{equation}
By \cite[Corollary 2.4]{DZ24}, $\chi(\widehat{H}) \left[ \widehat{H}  , \Lambda_V \right]$ is trace-class.  Therefore, we can take the trace on both sides of \eqref{eq:0s} indistinctly switch trace and integral. By cyclicity, we can then move one of the resolvents $(\widehat{H}-z)^{-1}$ around and obtain, after integration by parts:
\begin{align}
   \Tr \big( \chi(\widehat{H})  [g(\widehat{H}),\Lambda_V] \big) & = \int_{\C^+} \dd{\tg(z,\oz)}{\oz} (\widehat{H}-z)^{-1} \Tr\big( \chi(\widehat{H}) \left[ \widehat{H}  , \Lambda_V \right] (\widehat{H}-z)^{-1} \big) dm(z)
   \\
   & = \int_{\C^+} \dd{\tg(z,\oz)}{\oz} \Tr\big( (\widehat{H}-z)^{-2}  \chi(\widehat{H}) \left[ \widehat{H}  , \Lambda_V \right] \big) dm(z)
\end{align}
Switching trace and integral produces
   \begin{align}
   \Tr \big( \chi(\widehat{H})  [g(\widehat{H}),\Lambda_V] \big) & = \Tr \left( \int_{\C^+} \dd{\tg(z,\oz)}{\oz} \dd{(\widehat{H}-z)^{-1}}{z} \chi(\widehat{H}) \left[ \widehat{H}  , \Lambda_V \right] dm(z) \right)
   \\
   &  = \Tr \left( \int_{\C^+} \dd{\tg(z,\oz)}{z \p \oz} (\widehat{H}-z)^{-1} \chi(\widehat{H}) \left[ \widehat{H}  , \Lambda_V \right] dm(z) \right)
   \\
   & = \Tr \big(  g'(\widehat{H}) \chi(\widehat{H})[\widehat{H},\Lambda_V] \big).
\end{align}
This proves the first equality of \eqref{eq:0q}, and going back to \eqref{eq:0o}, the proof of \eqref{eq:0r}. \end{proof}


\appendix
\section{Locality estimates}
Here we collect some locality estimates.
\begin{lemma}\label{lem:commutator of switch with local is confined and local}
    If $A$ is a local Hamiltonian, then $[A,\Lambda_U]$ is local and decays away from $\partial U$.
\end{lemma}

\begin{proof}
    Assume that $A$ is local in the sense that \eql{\norm{A_xy}\leq \mu^{-1} \ee^{-\mu d(x,y)}} for some $\mu>0$ and some translation-invariant metric $d$.
    
    Writing $[A,\Lambda_U] = 2i\Im{\Lambda_U A \Lambda_{U^c}}$ we have the integral kernel estimate \eql{
    \norm{\left([A,\Lambda_U]\right)_{xy}}&\leq \frac{2}{\mu^2} \Lambda_U(x)\Lambda_{U^c}(y)   \ee^{-\mu d(x,y)}\\
    &\leq \frac{2}{\mu^2}   \ee^{-\mu \left(d(x,y)+d(x,U)+d(y,U^c)\right)}\,.
    }

    Now, using the triangle inequality we have $d(x,y)+d(x,U)\geq d(y,U)$ and hence
    \eql{
        d(x,y)+d(x,U)+d(y,U^c) \geq d(y,U)+d(y,U^c) \geq d(y,\partial U)
    } and similarly for $x$. As a result we find the estimate
    \eql{\label{eq:local and confined}
        \norm{\left([A,\Lambda_U]\right)_{xy}}&\leq \frac{2}{\mu^2}\exp\left(-\frac{1}{3}\mu \left(d(x,y)+d(x,\partial U)+d(y,\partial U) \right)\right)
    } which is what we were trying to prove.
\end{proof}
\begin{lemma}\label{lem:double commutator is trace-class}If $U,V$ are transverse and $P$ is local then $[\Lambda_U,P][\Lambda_V,P]$ is of trace class.
\end{lemma}
\begin{proof}
    We shall use \Cref{lem:commutator of switch with local is confined and local}. It says that since $P$ is local, the commutator $[\Lambda_U,P]$ is both local and decays away from $\partial U$. When we then take the product of two such commutators we get that there exists some $\mu$ (depending on the locality of $P$) such that
    \eql{
        \norm{\left([\Lambda_U,P][\Lambda_V,P]\right)_{xy}} &\leq \frac{1}{\mu}\ee^{-\mu\left(\norm{x-y}+d(x,\partial U)+d(x,\partial V)+d(y,\partial U)+d(y,\partial V)\right)}\qquad(x,y\in\Z^2)\,.
    } where $d$ is the Euclidean distance. Since $U,V$ are transverse, according to \Cref{def:transverse} there exists some $\alpha>0$ such that for large enough $x$ we have $\psi_{U,V}(x) \geq \norm{x}^{\alpha}$. Combining this with the estimate $\norm{A}_1 \leq \sum_{x,y\in\Z^2}\norm{A_{xy}}$ yields the desired result. 
\end{proof}

\begin{lemma}\label{lem:geometric index well defined}
    If $U,V$ are transversal and $P$ is a local projection then \eql{
    \left[\Lambda_V,\exp\left(2\pi i \Lambda_U P \Lambda_U\right)\right]\in\KK\,.
    }
\end{lemma}
\begin{proof}
    Given that $\exp(2\pi\ii x)-1 = h(x)(x^2-x)$ for some holomorphic function $h:\R\to\C$, and the fact that holomorphic function calculus preserves locality, together with \Cref{lem:commutator of switch with local is confined and local} which allows us to conclude locality and a decay away from $\partial V$, we merely need to establish that $\Lambda_U P \Lambda_U$ is a projection away from $\partial U$. However, the identity \eql{
        \left(\Lambda_U P \Lambda_U\right)^2-\Lambda_U P \Lambda_U = \Lambda_U P \Lambda_U P \Lambda_U - \Lambda_U P \Lambda_U = -\Lambda_U P \Lambda_{U}^\perp P \Lambda_U = -\left[\Lambda_U, P\right] \Lambda_{U}^\perp P \Lambda_U
    } makes it clear the difference is compact.
\end{proof}

\section{A geometric index theorem}\label{sec:geometric index theorem}
\begin{lemma}\label{lem:geometric inex thm}
    Let $U,V$ be two transverse sets. Then the geometric Hall conductance obeys the following index theorem 
    \eql{
        \NN^{U,V}(H) = 2\pi \sigma_{\mathrm{Hall}}^{U,V}(H)
    }
\end{lemma}
The proof follows precisely the same steps as if $U$ were the upper half plane and $V$ were the right half plane. It really only makes use of the transversal condition when passing to arbitrary $U,V$. It is included here for convenience of the reader. We note that usually the index theorem is proven for the Laughlin index, which makes the proof much longer.
\begin{proof}
    The proof of this is basically taken from \cite{KITAEV20062} but also appeared in \cite{FSSWY20}. 

    First, a short calculation shows that \eql{
    P\left[\left[\Lambda_{U},P\right],\left[\Lambda_{V},P\right]\right] = \left[P\Lambda_{U}P,P\Lambda_{V}P\right]\,.
    }
    Inserting this into the definition of $\sigma_{\mathrm{Hall}}^{U,V}(H)	$ in \eqref{eq:geometric Hall conductance} yields
    \eql{
    \frac{1}{\ii}\sigma_{\mathrm{Hall}}^{U,V}(H)	&=	\lctr\left(\left[P\Lambda_{U}P,P\Lambda_{V}P\right]\right) \\
	&=	\left(\frac{1}{2\pi}\int_{0}^{2\pi}d\alpha\right)\lctr\left(\left[P\Lambda_{U}P,P\Lambda_{V}P\right]\right)\\
	&=	\frac{1}{2\pi}\int_{0}^{2\pi}\lctr\left(\ee^{-\ii\alpha P\Lambda_{V}P}\left[P\Lambda_{U}P,P\Lambda_{V}P\right]\ee^{\ii\alpha P\Lambda_{V}P}\right)d\alpha \\
 & =	\frac{1}{2\pi\ii}\int_{0}^{2\pi}\lctr\left(\partial_{\alpha}\ee^{-\ii\alpha P\Lambda_{V}P}P\Lambda_{U}P\ee^{\ii\alpha P\Lambda_{V}P}\right)d\alpha}
 We now switch trace and integral to obtain:
	 \eql{
	\frac{1}{\ii}\sigma_{\mathrm{Hall}}^{U,V}(H) &=	\frac{1}{2\pi\ii}\lctr\int_{0}^{2\pi}\partial_{\alpha}\ee^{-\ii\alpha P\Lambda_{V}P}P\Lambda_{U}P\ee^{\ii\alpha P\Lambda_{V}P}d\alpha\\
	&=	\frac{1}{2\pi\ii}\lctr\left(\ee^{-\ii2\pi P\Lambda_{V}P}P\Lambda_{U}P\ee^{\ii2\pi P\Lambda_{V}P}-P\Lambda_{U}P\right)\\
	&=	\frac{1}{2\pi\ii}\lctr\left(\ee^{-\ii2\pi P\Lambda_{V}P}\underbrace{\left[P\Lambda_{U}P,\ee^{\ii2\pi P\Lambda_{V}P}\right]}_{P\left[\Lambda_{U},\ee^{\ii2\pi P\Lambda_{V}P}\right]P}\right)
	=	\frac{1}{2\pi\ii}\lctr\left(P\ee^{-\ii2\pi P\Lambda_{V}P}P\left[\Lambda_{U},\ee^{\ii2\pi P\Lambda_{V}P}\right]\right)\,.
    }

Next, note that \eql{
\ee^{-\ii2\pi P\Lambda_{U}P}
	&=	P^{\perp}+P\ee^{-\ii2\pi P\Lambda_{U}P}P\,.
    }
    However, 
    \eql{
    \lctr\left(P^{\perp}\left[\Lambda_{V},\ee^{\ii2\pi P\Lambda_{U}P}\right]\right)	&=	\lctr\left(P^{\perp}\left[\Lambda_{V},\ee^{\ii2\pi P\Lambda_{U}P}\right]P^{\perp}\right)\\
	&=	\lctr\left(P^{\perp}\Lambda_{V}\ee^{\ii2\pi P\Lambda_{U}P}P^{\perp}\right)-\lctr\left(P^{\perp}\ee^{\ii2\pi P\Lambda_{U}P}\Lambda_{V}P^{\perp}\right) = 0.
 }
    Hence, \eql{
       2\pi \cdot \sigma_{\mathrm{Hall}}^{U,V}(H)	=\lctr\left(\ee^{-\ii2\pi P\Lambda_{U}P}\left[\Lambda_{V},\ee^{\ii2\pi P\Lambda_{U}P}\right]\right)\,.
        }
    Now, this expression is of the form $\left(W^{\ast}\left[Q,W\right]\right)$
        where W is a unitary and Q is a projection, so we can finish using the well-known result of \cite{ASS94} about the index of a pair of projections.

        Note that to explain how to transition from $P\Lambda_{U}P$ to $\Lambda_{U}P\Lambda_{U}$ in the exponent, we proceed via \cite[Lemma 4.8]{Bols2023} and the identity \eql{P\Lambda_{U}P - \Lambda_{U}P\Lambda_{U} = \Lambda_{U}\left[P,\Lambda_{U}\right]P^\perp\,. }
        This completes the proof.
\end{proof}
\section{A geometric bulk-edge correspondence}\label{sec:geometric BEC}
\begin{theorem}\label{thm:geometric index BEC}
    We have \eql{\NN^{U,V}(H_+)-\NN^{U,V}(H_-) = \widehat{\NN}\left(\widehat{H},V\right)\,.}
\end{theorem}
The proof of this theorem already appeared in \cite{DZ24} without recourse to index theory. We provide here a sketch alternative to the tracial proof of \cite{DZ24} that follows the same lines as in \cite{FSSWY20}, but is actually simpler because we work with interface edge rather than sharply-terminated edge.  The interested reader can check details in \cite{FSSWY20}.

\begin{proof}[Sketch of proof]
    First it might be a good idea to establish that \eqref{eq:edge index formula} is well defined. This, however, follows clearly from the fact that since $g$ is smooth, using the Helffer-Sj\"ostrand formula, up to terms which are local and decay away from $\partial U$, we have 
    \eql{
        g(\widehat{H}) &\cong g(\Lambda_U H_+ \Lambda_U + \Lambda_{U^c} H_- \Lambda_{U^c})\\
        &\cong \Lambda_U g( H_+)\Lambda_U + \Lambda_{U^c} g(H_-) \Lambda_{U^c} = \Lambda_U P_+\Lambda_U + \Lambda_{U^c} P_- \Lambda_{U^c}\,.
    } The last equality follows thanks to the spectral gap condition, assuming that $\supp(g')$ is within the mutual gap of both $H_\pm$. 

    Now, since these two latter terms commute, their exponential is the product of the exponentials. Now using the logarithmic property of the Fredholm index (as well as stability up to compacts) we find 
    \eql{
        \widehat{\NN}\left(\widehat{H},V\right) &\equiv  \findex\left(\bbLambda_V \exp\left(2\pi\ii g(\widehat{H}) \right)\right) \\
        &= \findex\left(\bbLambda_V \exp\left(2\pi\ii \left(\Lambda_U P_+\Lambda_U + \Lambda_{U^c} P_- \Lambda_{U^c} \right)\right)\right) \\
        &= \findex\left( \bbLambda_V\left( \exp\left(2\pi\ii \Lambda_U P_+\Lambda_U \right)\exp\left(2\pi\ii \Lambda_{U^c} P_- \Lambda_{U^c} \right)\right)\right)\\
        &= \findex\left( \left(\bbLambda_V \exp\left(2\pi\ii \Lambda_U P_+\Lambda_U \right)\right)\left(\bbLambda_V\exp\left(2\pi\ii \Lambda_{U^c} P_- \Lambda_{U^c} \right)\right)\right) \\
        &= \findex\left(\bbLambda_V \exp\left(2\pi\ii \Lambda_U P_+\Lambda_U \right)\right)+\findex\left(\bbLambda_V\exp\left(2\pi\ii \Lambda_{U^c} P_- \Lambda_{U^c} \right)\right)\\
        &= \findex\left(\bbLambda_V \exp\left(2\pi\ii \Lambda_U P_+\Lambda_U \right)\right)-\findex\left(\bbLambda_V\exp\left(2\pi\ii \Lambda_{U} P_- \Lambda_{U} \right)\right)
    } which is what we wanted to prove.
\end{proof}

\section{A review of \eqref{eq:magic formula}}

In this section, we briefly review the origin of the formula \eqref{eq:magic formula}:
\begin{equation}\label{eq:8a}
    \sigma_b^{U,V}(H) = -i \Tr \big( P \big[ [P,\Lambda_U], [P,\Lambda_V] \big] \big) = \chi_{U,V} \cdot \sigma_\Hall(H), 
\end{equation}
where $U, V \subset \R^2$ are transverse sets and $\chi_{U,V}$ is their intersection number, reviewed in \S\ref{sec:existence of transversal set}. It relies on the following observations.

\noindent \textbf{1. The geometric bulk conductance is linear in $\Lambda_U, \Lambda_V$.} As a consequence, decomposing $U$ and $V$ in connected components produces a sum of geometric bulk conductances over connected sets. Therefore, we can assume that the sets $U$ and $V$ that are connected. 

\noindent \textbf{2. The geometric bulk conductance switches signs when replacing $U$ by $U^c$.} Combining this observation with Step 1, we can assume further that the sets $U^c$ and $V^c$ are connected. 

Brought together, Steps 1 and 2 show that the geometric bulk conductance associated to two general transverse sets is a weighted sum (with weights $\pm 1$) of geometric bulk conductances associated to two simple sets (connected sets whose complement is connected). In other words, we can assume that $U$ and $V$ are transverse simple sets.

\noindent \textbf{3. Boundaries of simple sets are connected.} Under the assumption that $U$ and $V$ are simple, $\p U$ and $\p V$ are the ranges of parametrized curves $\gamma_U$ and $\gamma_V$. We orient them so that $U$ and $V$ lie, respectively, to the left of $\gamma_U, \gamma_V$. In this setup, we define the intersection number $\chi_{U,V}$ as:
\begin{equation}
    \chi_{U,V} = \lim_{t \rightarrow +\infty} \1_V \circ \gamma_U(t) - \lim_{t \rightarrow -\infty} \1_V \circ \gamma_U(t).
 \end{equation}

The next two steps are analytic in nature; their execution demands more care -- with regard to uniformity in the parameters involved -- than is presented here.

\noindent \textbf{4. The geometric bulk conductance can be locally computed.} This is because the commutators $[P,\Lambda_U]$ are supported near $\p U$ and $\p V$, respectively. In particular, the commutator 
\begin{equation}
    \big[[P,\Lambda_U], [P,\Lambda_V] \big],
\end{equation}
is supported near $\p U \cap \p V$; this is a compact set. Hence, knowing $U$ and $V$ in a large enough ball is enough to compute $\sigma_b^{U,V}(H)$.

\noindent \textbf{5. The geometric bulk conductance is a robust quantity.} Specifically changing $U$ and $V$ in compact sets does not modify its value. Therefore, if $\chi_{U,V} = 1$, one can deform $U, V$ to (respectively) sets $U_n, V_n$ that look like (respectively) the upper half-plane and the right-half plane in the ball centered $\Bb(0,n)$ at $0$, of radius $n$ (if $\chi_{U,V} = -1$, it suffices to switch these half-planes).  By Step 4, it follows that 
\begin{equation}
    \sigma_b^{U,V}(H) = \sigma_b^{U_n,V_n}(H) = \sigma_b^{\{x_2>0\},\{x_1 > 0\}}(H) + o(1) = \chi_{U,V} \cdot \sigma_\Hall(H) + o(1), \qquad n \rightarrow \infty.
\end{equation}
Taking the limit as $n \rightarrow \infty$ produces the formula \eqref{eq:8a}.

\bibliographystyle{amsxport}
\bibliography{ref.bib}

\end{document}